\documentclass[a4paper]{article}

\usepackage[utf8]{inputenc}
\usepackage[english]{babel}
\usepackage{amsmath}
\usepackage{amsthm}
\usepackage{amssymb}
\usepackage{thmtools}
\usepackage{thm-restate}
\usepackage{algorithm}
\usepackage{algorithmicx}
\usepackage[noend]{algpseudocode}

\usepackage{fullpage}
\usepackage[pdftex]{graphicx}
\usepackage{hyperref}

\newtheorem{theorem}{Theorem}[section]

\newtheorem{lemma}[theorem]{Lemma}

\usepackage{tikz}

\usepackage{caption}
\usepackage{subcaption}
\title{An Improved PTAS for Covering Targets with Mobile Sensors}

\newcommand{\pA}{\mathcal{A}}
\newcommand{\pB}{\mathcal{B}}

\author{Nonthaphat Wongwattanakij\thanks{E-mail: nonthaphat.wo@ku.th. Department of Computer Engineering, Kasetsart University, Bangkok, Thailand.  Supported by Graduate Student Research Grant, Faculty of Engineering, Kasetsart University.}
  \and Nattawut Phetmak\thanks{E-mail: nattawut.p@ku.th. Department of Computer Engineering, Kasetsart University, Bangkok, Thailand.} 
  \and Chaiporn Jaikaeo\thanks{E-mail: chaiporn.j@ku.ac.th. Department of Computer Engineering, Kasetsart University, Bangkok, Thailand.} 
  \and Jittat Fakcharoenphol\thanks{E-mail: jittat@gmail.com. Department of Computer Engineering, Kasetsart University, Bangkok, Thailand.  Supported by Thailand Reseach Fund, Grant RSA-6180074} 
}

\begin{document}

\maketitle

\begin{abstract}
  This paper considers a movement minimization problem for mobile sensors.
  Given a set of $n$ point targets, the {\em $k$-Sink Minimum Movement Target Coverage Problem} is to schedule mobile sensors, initially located at $k$ base stations, to cover all targets minimizing the total moving distance of the sensors.  We present a polynomial-time approximation scheme for finding a $(1+\epsilon)$ approximate solution running in time $n^{O(1/\epsilon)}$ for this problem when $k$, the number of base stations, is constant.  Our algorithm improves the running time exponentially from the previous work that runs in time $n^{O(1/\epsilon^2)}$, without any target distribution assumption.  To devise a faster algorithm, we prove a stronger bound on the number of sensors in any unit area in the optimal solution and employ a more refined dynamic programming algorithm whose complexity depends only on the width of the problem.


\end{abstract}


\section{Introduction}
Wireless sensor networks (WSNs) are used in various types of applications such as military, industry, and disaster responses~\cite{Wireless2002Akyildiz,Wireless2008Yick}.
To effectively monitor a region of interest, sensor nodes are placed at specific locations so that they can collectively provide overall sensing coverage for the entire area~\cite{Coverage2011Wang,WuZLD-coverage-book}.
WSNs may consist of static sensor nodes, mobile sensor nodes, or both.
With presence of mobile sensor nodes, monitor areas may be dynamically adjusted by having these mobile nodes navigate to appropriate positions.
This capability helps reducing the overall number of sensor nodes and/or provides better sensing coverage over regions of interest.

In this paper we study the $k$-Sink Minimum Movement Target Coverage Problem, where we would like to schedule sensors, with sensing range $r$, from $k$ stations to cover all targets, while minimizing the total distance for every sensor.  
Gao, Chen, Wu, and Chen~\cite{GaoCWC-journalTNET} proposed a polynomial-time approximation scheme for this problem that runs in time $n^{O(1/\epsilon^2)}$ and returns a solution whose cost is at most $(1+\epsilon)$ of the optimal cost when the targets are distributed uniformly and randomly in the surveillance region.  This result first appeared as~\cite{ChenGWC_INFOCOM16} in INFOCOM'16.

Our contribution is as follows.  We observe that, as in many PTAS's for other geometric problems, their general framework, based on exhaustive search and shifting, does not depend on the target distribution; thus it works without the distribution assumption.  More over, we notice that the running time bound of their algorithm depends on an erroneous claim of the bound on the number of sensors needed in the optimal solution (in Theorem~3 in~\cite{ChenGWC_INFOCOM16} and in Theorem~5 in~\cite{GaoCWC-journalTNET}) for a particular region (see Appendix~\ref{sect:counter}) and provide a correction.  Our proof actually provides a stronger guarantee, allowing us to devise a more refined dynamic programming subroutine, the main contribution of this work, used in the algorithm that runs in time $n^{O(1/\epsilon)}$, which is an exponential improvement on the running time.  While we consider our result a fairly theoretical one, we hope that the dynamic programming approach presented here might find more applications in other networking problems.

We describe the network model, problem statements and a review of Gao~{\em et al.}~\cite{GaoCWC-journalTNET,ChenGWC_INFOCOM16} algorithm, called Energy Effective Movement Algorithm (EEMA), in Section~\ref{sect:prelims}.  Section~\ref{sect:algo} describes the improved algorithm and proves the running time and its performance guarantee.  Our running time analysis depends on the bound on the sensors needed in the optimal solution, proved in Section~\ref{sect:proofs}.

\section{Related works}

Coverage problems has been widely studied in computational geometry~\cite{art-gallery}.  These problems are important to wireless sensor networks (see recent survey in~\cite{Coverage2011Wang}
and the book by Wu~{\em et al.}~\cite{WuZLD-coverage-book}).  There are numerous heuristic algorithms for the placement and deployment problem, e.g., the algorithms for full coverage based on packing has been proposed by Wang and Tseng~\cite{WangT08-without-obstacle} for areas without obstacles and by Wang, Hu, and Tseng~\cite{WangHT08-with-obstacle} for areas with obstacles, for coverage enhancement using virtual forces~\cite{ZouC04-virtualforce}, for strategies that deals with coverage holes using Voronoi polygons~\cite{MahboubiMASM14-voronoi}, and for sensor relocation using Delaunay triangulation~\cite{KhampeerpatJ17}.  This paper as well as that of Gao~{\em et al.}~\cite{GaoCWC-journalTNET} studies movement scheduling problems.  For dynamic area coverage, Liu, Dousse, Nain, and Towsley~\cite{LiuDNT13-dynamic-coverage} considered coverage using mobile sensor networks.  Liu~{\em et al}~\cite{LiuDWS08-barrier} and He~{\em et al.}~\cite{HeCZZ12-cost-effective-barrier} considered scheduling problem when the goal is to protect an area using a barrier of sensors.  Ammari~\cite{Ammari12} tried to minimize sensor movement cost while ensuring area coverage.   For strategies for improving coverage, Wang, Lim, and Ma~\cite{WangLM09-survey-improving} provided a survey.  Guo and Jafarkhani~\cite{GuoJ19-movement-efficient} presented centralized and decentralized heuristics for optimizing sensor movement under network lifetime constraints.  Elhoseny~{\em et al.}~\cite{ElhosenyTYH18-opt-k-coverage} employed the genetic algorithm as a heuristic to deal with coverage requirements.

When theoretical guarantees are needed, approximation algorithms for many versions of coverage problems have been studied.  Section~\ref{sect:prelims-ptas} describes related work on polynomial-time approximation schemes on geometric graphs closely related to this work.  In this section, we list a few notable works on approximation algorithms for coverage problems.  For target coverage, if the goal is to schedule sensor activation groups to maximize the life time, Ding~{\em et al.}~\cite{DingWWLZW12-const-lifetime} devised a $(4+\xi)$-approximation algorithm for any $\xi>0$.  For restricted form of $k$-coverage, Xu and Song~\cite{XuS14-restricted} gave a $3$-approximation algorithm.   When a threshold for movement of each sensor is given and the objective is to maximize target coverage, a recent result by Liang, Shen, and Chen~\cite{LiangSC21-max-target} gave an algorithm with an approximation ratio of $(1-1/e)$.   Another important class of coverage problem is to find dominating sets in unit disk graphs.  Marathe~{\em et al.}~\cite{MaratheBHRR95} described a 5-approximation algorithm for finding dominating sets and a 10-approximation algorithm for finding connected dominating sets.  For weighed dominating sets, Erlebach and Mihal\'{a}k~\cite{ErlebachM10-MWDS} and Zou~{\em et al.}~\cite{ZouWXLDWW2011-mwcds} independently devised a $(4+\epsilon)$-approximation algorithms.  Zou~{\em et al.}~\cite{ZouWXLDWW2011-mwcds} also present a $(5+\epsilon)$-approximation algorithm for the connected version of the problem.  A recent book by Wu~{\em et al.}~\cite{WuZLD-coverage-book} is an excellent source for these developments.


\section{Preliminary and Problem Statement}
\label{sect:prelims}

\subsection{Model and Problem Statement}

Following Gao {\em et al.}~\cite{GaoCWC-journalTNET,ChenGWC_INFOCOM16}, the homogeneous network considered in this paper can be modeled as $G=(T,P,r,S)$, where $T=\{t_1,t_2,\ldots,t_n\}$ is the target set, $P=\{p_1,p_2,\ldots,p_k\}$ is the station set, $r$ is the sensing range, and $S=\{s_1,s_2,\ldots,s_m\}$ is the mobile sensor set.  The targets are static points on the plane that we want to cover using sensors with sensing range $r$.  We say that a sensor set $S$ {\em covers} a point $t_i$ if there exists a sensor $s_j\in S$ such that $d(s_j,t_i)\leq r$, where $d(a,b)$ denotes the Euclidean distance between two points $a$ and $b$ on the plane.  The coverage requirement ensures that every target $t_i\in T$ must be covered.  Each sensor $s_j\in S$ must be scheduled to move to its location from some station $p_\ell\in P$; the movement distance is $d(s_j,p_\ell)$.  Since the station set $P$ is static we can easily find $p_\ell\in P$ that minimizes $d(s_j,p_\ell)$.  We would later refer to the distance just as $d(s_j)$ defined to be $min_{p_\ell\in P} d(s_j,p_\ell)$.

Gao~{\em et al.}~\cite{GaoCWC-journalTNET,ChenGWC_INFOCOM16} defined the {\em $k$-Sink Minimum Movement Target Coverage Problem} ($k$-MMTC) as follows.  Given a set of targets $T$, station set $P$, and sensing radius $r$, find a set of sensors $S$ covering $T$ such that the total movement distance is minimized.  They proposed an algorithm called Energy Effective Movement Algorithm (EEMA) that is claimed to be a PTAS that solves this problem (reviewed in Subsection~\ref{sect:review-eema}).  

We would like to note that Gao~{\em et al.}~\cite{GaoCWC-journalTNET,ChenGWC_INFOCOM16} also assumed that the targets are distributed uniformly and randomly in the surveillance region.  Our work does not need this requirement.  In fact, we observe that their algorithm and proofs work even without this distribution assumption as well.

\subsection{Polynomial-Time Approximation Scheme}
\label{sect:prelims-ptas}

A {\em polynomial-time approximation scheme} (PTAS) is an algorithm that, given a constant $\epsilon>0$, returns a solution to the problem of cost within $(1+\epsilon)$ of the optimal in time polynomial in $n$, the size of the problem, assuming that $\epsilon$ is a constant.  Note that the running time may depend exponentially on $1/\epsilon$, i.e., the running time might increase very rapidly as the solution quality increases (as $1/\epsilon\rightarrow\infty$ when $\epsilon\rightarrow 0$).

A general framework for developing PTAS for geometric problems introduced by Hunt~{\em et al.}~\cite{HUNT1998238} follows the approach used in planar graphs by Baker~\cite{Baker94}.  First the problem is decomposed into much smaller subproblems which can be solved separately exhaustively either by dynamic programming or brute-force search.  Then, the ``locally-optimal'' solutions from these subproblems are combined to obtain the final solution.  Clearly, the combined solution might not be globally optimal because the dependency between each subproblem is neglected.  However, one can typically ensure that the neglected cost is small, i.e., at most $1/\epsilon$ fractions of the optimal cost, when the size of each subproblem is large enough using an averaging argument.  For geometric problems mostly on unit disk graphs, Hunt~{\em et al.}~\cite{HUNT1998238} consider an area of size $O(1/\epsilon)\times O(1/\epsilon)$ and use the fact that there is a feasible solution of size $O(1/\epsilon^2)$ to devise many PTAS's that run in time $n^{O(1/\epsilon^2)}$.  For dominating sets, which is another form of coverage problems, if the nodes are unweighted, there are several PTAS for the standard problem~\cite{HUNT1998238,nieberg2005ptas} and for the case when the dominating set is required to be connected~\cite{ZhangGWD2009unitball}.  (See~\cite{WuZLD-coverage-book} for a comprehensive treatment of this topics.)

The dynamic programming approach for covering problems is widely used.  Erik Jan van Leeuwen~\cite{vanLeeuwen05,vanLeeuwen06,vanLeeuwen-thesis09} has employed roughly the same speed-up idea to the dominating set problem in unit disk graphs.

\subsection{Review of EEMA~\cite{GaoCWC-journalTNET,ChenGWC_INFOCOM16}}
\label{sect:review-eema}

EEMA proposed by Gao, Chen, Wu, and Chen~\cite{GaoCWC-journalTNET,ChenGWC_INFOCOM16} works in two phases.  The first phase divides the surveillance region into a set of subareas.  This changes the nature of the problem from continuous optimization problem to a combinatorial one.  The second phase approximates the solution to the combinatorial problem constructed in the first phase.  We shall describe both phases and our contribution is essentially on the second phase of their proposed EEMA.

\subsubsection{First phase of EEMA}  

For each target $t_i$, the {\em detection cycle} $D(t_i)$ is a circle of radius $r$ centered at $t_i$.  Note that a sensor in $D(t_i)$ covers $t_i$.  Detection cycles may intersect; thus, a single sensor in the intersection can detect all targets associated with these cycles.  The first phase of EEMA identifies all minimal subareas defined by intersections of detection cycles (see Figure~\ref{fig:detection}).  Each subarea is called a {\em curved-boundary.}  Gao~{\em et al.} proved the following lemma (implicitly in Theorem~2 in~\cite{GaoCWC-journalTNET}).

\begin{figure}
  \centering
  \begin{tikzpicture}
    \coordinate (T1) at (1.0,2.1);
    \coordinate (T2) at (2.1,1.0);
    \coordinate (T3) at (2.7,2.2);
    \coordinate (T4) at (4.2,2.6);
    \draw (T1) circle (1);
    \draw (T2) circle (1);
    \draw (T3) circle (1);
    \draw (T4) circle (1);
    \filldraw[black] (T1) circle (1.5pt) node[anchor=south east] {$t_1$};
    \filldraw[black] (T2) circle (1.5pt) node[anchor=north] {$t_2$};
    \filldraw[black] (T3) circle (1.5pt) node[anchor=south] {$t_3$};
    \filldraw[black] (T4) circle (1.5pt) node[anchor=south west] {$t_4$};
  \end{tikzpicture}
  \caption{Detection cycles for $t_1,t_2,t_3,$ and $t_4$.  There are 9 minimal subareas to be considered.}
  \label{fig:detection}
\end{figure}
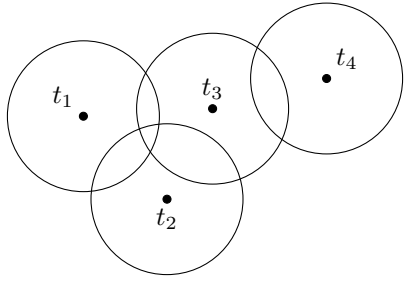

\begin{lemma}[\cite{GaoCWC-journalTNET}]
The number of minimal subareas defined by intersections of all $n$ detection cycles is $O(n^2)$
\end{lemma}

They also gave an $O(n^2)$-time algorithm (called the Graph Conversion algorithm) for finding all minimal subareas.

\subsubsection{Second phase of EEMA}
During the second phase, Gao~{\em et al.} noted that placing a sensor anywhere in a particular minimal subarea covers the same set of target, defined as {\em covered-set} of that subarea.  Therefore, to minimize the total distance, if one has to schedule a sensor to that subarea, the sensor must travel from the closest station.  For the subarea $\sigma_i$, they gave an algorithm for finding $w(\sigma_i)$, the minimum distance from any point in $\sigma_i$ to any station in $P$.  Given the representation of all subareas, computing $w(\sigma_i)$ for all subareas $\sigma_i$ takes only $O(n^2k)$.

After all these preprocessing steps, we are left with a combinatorial problem of choosing a subset of subareas to cover all targets.

Let $Q$ be the bounding box containing all targets.  As in many geometric PTAS, Gao~{\em et al.} then employed the algorithm, called the Partition algorithm, that divides $Q$ into cells $cell(Q)$, each of size $2mr\times 2mr$, solves the selection optimally for each cell $e\in cell(Q)$, and returns the union of selected subareas.  

The essential subproblems for the Partition algorithm deals with each cell $e$.  To solve each subproblem, Gao~{\em et al.} used a brute-force search that runs in time $n^{O(m^2)}$.  To bound the running time they proved the following key lemma (implicitly in Theorem~5).

\begin{lemma}[\cite{GaoCWC-journalTNET,ChenGWC_INFOCOM16}]
The number of sensors (and subareas selected) in the optimal solution in each cell $e$ is at most $O(m^2)$.
\end{lemma}

However, this lemma is incorrect (see Appendix~\ref{sect:counter}).  We provide a fix to this lemma (in Section~\ref{sect:proofs}) that implies that the number of sensors is $O(m^2+k)$, yielding that their algorithm for each subproblem actually runs in time $n^{O(m^2+k)}$.  More generally, if we have an upperbound of $L$ on the number of sensors, a brute-force algorithm would run in time $n^{O(L)}$.  We also provide a stronger bound that allows us to improve the running time to $n^{O(m+k)}$, improving the running time exponentially.

After solving each subproblem, each solution is combined straight-forwardly.  Gao~{\em et al.} used shifting techniques to show that if the targets are distributed uniformly in $Q$ (as described in the model in Section~III-A), the total cost is $(1+3/m)$ times the optimal cost, providing the guarantee of $1+\epsilon$ of the PTAS.  We note that their proof works even without the assumption on the distribution of the targets, for completeness, we also provide the proof of this performance guarantee in Section~\ref{sect:performance}.


\section{The Algorithm}
\label{sect:algo}

We focus on the second phase of EEMA where we work entirely on the combinatorial version of the problem (after the transformation using the Graph Conversion algorithm of~\cite{GaoCWC-journalTNET,ChenGWC_INFOCOM16}).  Recall that we are given a set of $h$ minimal subareas $R=\{\sigma_1,\sigma_2,\ldots,\sigma_h\}$ defined by intersections of detection cycles.  Recall also that $h=O(n^2)$.  Also for each subarea $\sigma_i$, we also know the minimum distance $w(\sigma_i)$ needed if we schedule a sensor into this subarea.  The goal is to select a subset of subareas $R'\subseteq R$ that minimizes the total distance while ensuring that these subareas cover all targets.

Our main subroutine works with a square region of size $2mr\times 2mr$, referred to as a {\em cell}.  In Subsection~\ref{sect:dp}, we present our key contribution, an algorithm for solving the problem optimally when all targets are inside a cell that runs in time $n_e^{O(m+k)}$, where $n_e$ is the number of subareas that can cover targets in $e$.  Since later on we set $m=\Theta(1/\epsilon)$, this implies the total running time of $n^{O(1/\epsilon)}$, an improvement over the brute-force algorithm of~\cite{GaoCWC-journalTNET} that runs in time $n^{O(1/\epsilon^2)}$.

To solve the whole problem, the surveillance region is partitioned into many cells, where the dynamic programming can be applied to.  Let $Q$ be the bounding box covering the surveillance region.  More specifically, assume that the lower-left corner of $Q$ is at $(0,0)$ and its upper-right corner is at $(M,M)$.  To facilitate the shifting, we further assume that targets are inside the square whose corners are $(2mr,2mr)$ and $(M,M)$, i.e., there are no targets whose $x$ or $y$ co-ordinates are less than $2mr$.  We divide $Q$ uniformly into square cells of size $2mr\times 2mr$.   There will be $\lceil M/2mr \rceil\cdot\lceil M/2mr \rceil$ cells.  Each cell is a square whose lower-left and upper-right corners are of the form $(i\cdot 2mr, j\cdot 2mr)$ and $((i+1)\cdot 2mr, (j+1)\cdot 2mr)$.  Let $cell(Q)$ be the set of all cells.  See Figure~\ref{fig:bounding-noshift} (a).

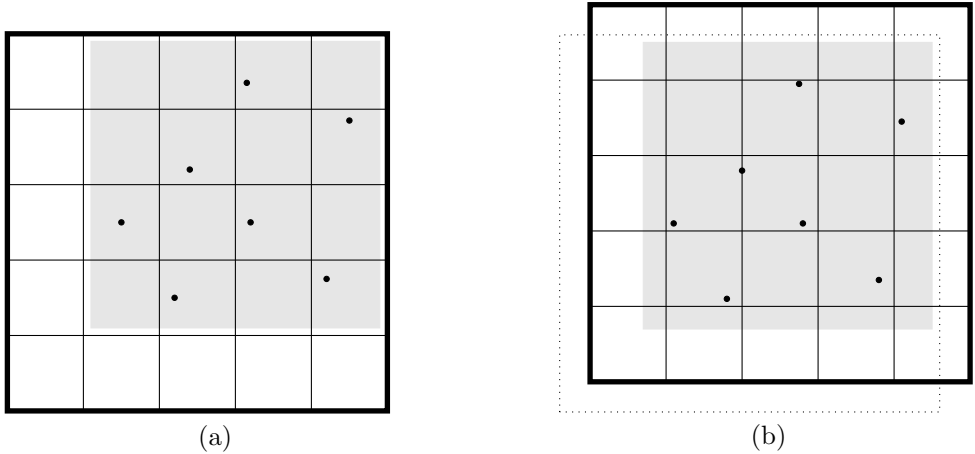
\begin{figure}
  \centering
  \begin{minipage}{.45\textwidth}  
  \centering 
  \begin{tikzpicture}
    \filldraw[white] (0,0) rectangle (5.5,5.5);
    \filldraw[gray!20] (1.1, 1.1) rectangle (4.9,4.9);
    \foreach \x/\y in {1.5/2.5, 2.2/1.5, 4.2/1.75, 3.2/2.5, 4.5/3.85, 2.4/3.2,  3.15/4.35} {
        \filldraw (\x,\y) circle (1pt);
    }
    \draw[line width=2pt] (0,0) rectangle (5,5);
    \foreach \x in {1,...,4} {
        \draw (0,\x) -- (5,\x);    
        \draw (\x,0) -- (\x,5);    
    }
  \end{tikzpicture} \\ (a)
  \end{minipage}
  \begin{minipage}{.45\textwidth}
  \centering
  \begin{tikzpicture}
    \filldraw[white] (0,0) rectangle (5.5,5.5);
    \draw[dotted] (0, 0) rectangle (5,5);
    \filldraw[gray!20] (1.1, 1.1) rectangle (4.9,4.9);
    \foreach \x/\y in {1.5/2.5, 2.2/1.5, 4.2/1.75, 3.2/2.5, 4.5/3.85, 2.4/3.2,  3.15/4.35} {
        \filldraw (\x,\y) circle (1pt);
    }
    \draw[line width=2pt] (0.4,0.4) rectangle (5.4,5.4);
    \foreach \x in {1,...,4} {
        \draw (0.4,\x+0.4) -- (5.4,\x+0.4);    
        \draw (\x+0.4,0.4) -- (\x+0.4,5.4);    
    }
  \end{tikzpicture} \\ (b)
  \end{minipage}  
  \caption{(a) the input region with targets (shown in shaded area), the bounding box $Q$, and the set of all cells of size $2mr\times 2mr$.
  (b) The shifted bounding box $Q_f$ and its cells $cell(Q_f)$.
  }
  \label{fig:bounding-noshift}
\end{figure}

The algorithm then applies the dynamic programming subroutine to every cell $e\in cell(Q)$, and returns the union of the selected subareas.  Gao~{\em et al.}~\cite{GaoCWC-journalTNET} proved that the solution returned is $4$-approximate.  

Using standard technique for PTAS design, as also in~\cite{GaoCWC-journalTNET}, we shift the bounding box to ``average out'' the cost of the optimal to obtain the $(1+\epsilon)$-approximate solution.  We discuss this step in Subsection~\ref{sect:performance} and present the proof of the performance guarantee. 

Finally, using the bounds to be proved in Section~\ref{sect:proofs}, we give the analysis of the running time in Subsection~\ref{sect:running-time}.

\subsection{Dynamic Programming}
\label{sect:dp}

In this section, we describe a dynamic programming that, for a given cell $e$ of size $2mr\times 2mr$, finds the optimal schedule for sensors to cover all the targets in $e$.  The cell $e$ is partitioned into vertical strips each of size $r\times 2mr$ (defined formally later).  
If we know that the number of sensors in the optimal solution for each strip of $e$ is at most $L$, the dynamic programming algorithm in this section runs in time $n_e^{O(L)}$, where $n_e$ is the number of subareas containing any targets inside cell $e$.
We give a proof in Lemma~\ref{lemma:sensor-bound} that $L=O(m+k)$; since we assume that $k$ is a constant, the running time of the dynamic programming is $n_e^{O(m)}$.

Because we use dynamic programming, we shall define other types of subproblems.  To avoid confusion later in this section, we shall refer to this main subproblem in a cell as the {\em cell problem}.    We further divide the cell into $2m$ vertical strips, each of size $r\times 2mr$.  We refer to these strips as strips $H_1, H_2,\ldots, H_{2m}$.  See Figure~\ref{fig:dp-subproblems}.

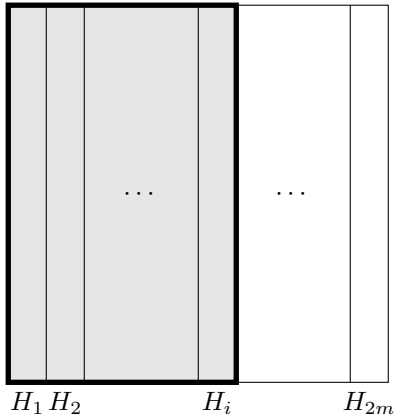
\begin{figure}
  \centering
  \begin{tikzpicture}
    \filldraw[gray!20] (0.0,0) rectangle (3,5);
    \draw[line width=2pt] (0,0) rectangle (3,5);
    \draw (0,0) rectangle (5,5);
    \draw (0.5,0) -- (0.5,5);
    \draw (1.0,0) -- (1.0,5);
    \draw (2.5,0) -- (2.5,5);
    \draw (3.0,0) -- (3.0,5);
    \draw (4.5,0) -- (4.5,5);
    \node[anchor=north] at (0.25,0) {$H_1$};
    \node[anchor=north] at (0.75,0) {$H_2$};
    \node[anchor=north] at (2.75,0) {$H_i$};
    \node[anchor=north] at (4.75,0) {$H_{2m}$};
    \node at (1.75,2.5) {$\cdots$};
    \node at (3.75,2.5) {$\cdots$};
  \end{tikzpicture}
  \caption{Strips $H_1,\ldots,H_{2m}$ and subproblem $A_i$ (shown in shade).}
  \label{fig:dp-subproblems}
\end{figure}

We define the following subproblems $\pA_i$ for each $1\leq i\leq 2m$ as follows.

{\bf Subproblem $\pA_i$:} find subset of subareas $R'$ that covers all targets in $H_1\cup H_2\cup\cdots\cup H_i$ with the minimum total distance.

Let $a_i$ denote the optimal total distance for $\pA_i$.  Recall that with this definition $a_{2m}$ is the optimal cost for the cell problem.  Here we only consider the algorithm that finds the minimum cost; to also find the actual solution, one only need to keep the best solutions together with its costs in the dynamic programming states.

For strip $H_i$, let $R_i\subseteq R$ be subset of subareas that cover some target in $H_i$.
For a subset $U\subseteq R_i$ of subareas of size at most $L$ covering some targets in $H_i$, we also define subproblems $\pB_i(U)$ for each $1\leq i\leq 2m$ as follows.

{\bf Subproblem $\pB_i(U)$:} find subset of subareas $R'$ that covers all targets in $H_1\cup H_2\cup\cdots\cup H_i$ such that $R'\cap R_i=U$ with the minimum total distance.

Note that $\pB_i(U)$ considers the same problem as $\pA_i$ but with additional conditions that to cover targets in $H_i$ one must use subareas in $U$.  
Let $b_{i,U}$ denote the optimal total distance for $\pB_i(U)$.  If it is impossible to do so, we let $b_{i,U}=\infty$.  Therefore, we know that 
\[
a_i = \min_{U\subseteq R_i: |U|\leq L} b_{i,U}.
\]
Hence, it is enough to solve $\pB_i(U)$.

As the base case, to solve $\pB_1(U)$ for each $U\subseteq R_1$, we can easily check if $U$ covers all targets and keep its cost or keep the value of $\infty$ otherwise.  Because $|U|\leq L$, there are at most $O(n_e^{L})$ subproblems to consider.  Since it takes $O(n^3)$ for each subproblem, the total running time is $n_e^{O(L)}$.

To solve $\pB_i(U)$ for each $U\subseteq R_i$, first note that since the strip width is $r$, any subareas $\sigma'$ covering any targets in $H_{i-2}$ cannot be in $R_i$, i.e., they cannot help covering targets in $H_i$.  On the other hand, subareas in $R_i$ cannot cover any targets in $H_1,\ldots,H_{i-2}$.  Therefore, when considering solutions, we only need to worry about dependencies between subareas in $U\in R_i$ and $U'\in R_{i-1}$.  Thus to find the solution for $\pB_i(U)$, we enumerate all $U'\in R_{i-1}$ and checks 
\begin{itemize}
    \item if $U$ and $U'$ are compatible (i.e., for subareas $\sigma$ in both $R_{i-1}$ and $R_i$, $U$ and $U'$ make the same decision, i.e., $\sigma\in U$ iff $\sigma\in U'$), and
    \item if $U\cup U'$ covers all targets in $H_i$.
\end{itemize}
If both conditions are met, the set of subareas for $\pB_{i-1}(U')$ and $U$ form a feasible solution for $\pB_i(U)$.  Finally, we choose $U'$ that minimizes the cost.

For each $U$, there are $O(n_e^{L})$ subsets $U'$ to consider.  Therefore to solve all $\pB_i(U)$ for all $U\in R_i$, the algorithm considers $O(n_e^{L})\cdot O(n_e^{L})=n_e^{O(L)}$, resulting in the running time of $n_e^{O(L)}$ as well.  The total running time for solving all $\pA_i$ is thus $m\cdot n_e^{O(L)}=m\cdot n_e^{O(m+k)}$, since $L=O(m+k)$.

\subsection{Shifting and the Performance Guarantee}
\label{sect:performance}

The shifting procedure proceeds in $m$ rounds.  For each round $f$, for $0\leq f\leq m-1$, the bounding box $Q$ is translated so that its lower-left corner is at $(2f,2f)$.  Let $Q_f$ denote the translated bounding box and $cell(Q_f)$ be the cells in $Q_f$.  Let $S_f$ be the set of subareas returned from round $f$.  The shifting procedure selects the solution with the minimum cost over returned solutions $S_f$'s.
See Figure~\ref{fig:bounding-noshift} (b).

We shall prove that there exists round $f$ such that the cost of $S_f$ is $1+4/m$ of the optimal cost.  Note that Gao~{\em et al.} also presented a similar proof in Theorem~8 of~\cite{GaoCWC-journalTNET} (appeared as Theorem~6 in~\cite{ChenGWC_INFOCOM16}) but they did not realize that this proof works for general case without the target distribution assumption.  We include the proof here for completeness.

\begin{theorem}
The shifting procedure returns a solution whose cost is at most $1+4/m$ of the optimal cost.
\label{thm:perf}
\end{theorem}
\begin{proof}
For any set of subareas $S'$, let $w(S')$ be its cost, i.e,
\[
w(S') = \sum_{\sigma\in S'} w(\sigma).
\]

Let $S^*$ be the optimal set of subareas and $w(S^*)$ be the optimal total movement cost. 
Consider round $f$ and the returned solution $S_f$.  For each $e\in cell(Q_f)$, for a set of subareas $S'$, let $S'(e)$ be the set of subareas in $S'$ that cover some target in $e$.
From the optimality of the dynamic programming algorithm, we have that for any $e$,
\[
\sum_{\sigma\in S_f(e)} w(\sigma) \leq \sum_{\sigma\in S^*(e)} w(\sigma).
\]
For a cell $e$ and a set of subareas $S'$, let $S'(\delta(e))$ be the set of subareas in $S'$ that cover some target in $e$ and also intersect the boundary of $e$.  We can split the cost of $S^*(e)$ into two parts, i.e.,
\[
\sum_{\sigma\in S^*(e)} w(\sigma)=
\sum_{\sigma\in S^*(e) \setminus S^*(\delta(e))} w(\sigma) +
\sum_{\sigma\in S^*(\delta(e))} w(\sigma).
\]
Note that for different $e,e'\in cell(Q_f)$, $S^*(e) \setminus S^*(\delta(e))$ and $S^*(e') \setminus S^*(\delta(e'))$ are disjoint.
Thus summing the cost for $S_f$, we have that $w(S_f)$ equals
\begin{align*}
\sum_{e}\sum_{\sigma\in S_f(e)} w(\sigma) &\leq
\sum_{e}\sum_{\sigma\in S^*(e)} w(\sigma) \\
&=
\sum_{e}\sum_{\sigma\in S^*(e) \setminus S^*(\delta(e))} w(\sigma) \;  + \\
&
\; \; \; \; \sum_{e}\sum_{\sigma\in S^*(\delta(e))} w(\sigma), \\
&\leq
w(S^*) + \sum_{e}\sum_{\sigma\in S^*(\delta(e))} w(\sigma),
\end{align*}
where $e$ is summed over $e\in cell(Q_f)$, and the last inequality follows from the disjointness of subareas strictly inside a cell $e$.

Taking the average over $w(S_f)$ for all $f$, we get
\begin{align*}
\frac{1}{m}\sum_{f=0}^{m-1} w(S_f)
&\leq
\frac{1}{m}\sum_{f=0}^{m-1}\left(w(S^*) + \sum_{e\in cell(Q_f)}\sum_{\sigma\in S^*(\delta(e))} w(\sigma)\right)\\
&= w(S^*) + \frac{1}{m}\sum_{f=0}^{m-1}\left(\sum_{e\in cell(Q_f)}\sum_{\sigma\in S^*(\delta(e))} w(\sigma)\right).
\end{align*}
Consider a particular subarea $\sigma\in S^*$.  Note that there are at most $4$ cells $e$ over all shifted instances such that $\sigma\in S^*(\delta(e))$, because $\sigma$ may intersect at most one vertical cell boundary twice and another one horizontal cell boundary twice.  This implies that the second term is at most
\[
\sum_{f=0}^{m-1}\left(\sum_{e\in cell(Q_f)}\sum_{\sigma\in S^*(\delta(e))} w(\sigma)\right)\leq
4w(S^*).
\]
Plugging this inequality into the previous bound, we have that
\[
\frac{1}{m}\sum_{f=0}^{m-1} w(S_f)\leq w(S^*)+\frac{4}{m} w(S^*)=(1+4/m) w(S^*).
\]
Since all $w(S_f)$'s are non-negative, at least one $f$ exists such that $w(S_f)\leq (1+4/m)w(S^*)$, as required.
\end{proof}

\subsection{Running time}
\label{sect:running-time}

From Theorem~\ref{thm:perf}, we need to set $m=4/\epsilon$ to guarantee that the algorithm gives a $(1+\epsilon)$-approximate solution.  For each shifted bounding box $Q_f$ and for each cell $e\in cell(Q_f)$, the dynamic programming algorithm from Section~\ref{sect:dp} runs in time $m\cdot n_e^{O(L)}$ where $n_e$ is the number of subareas covering some target in $e$ and $L$ is the upperbound on the number of sensors needed.

In Section~\ref{sect:proofs}, we show that $L=O(m+k)=O(m)$ since we assume that $k$ is a constant.  Therefore, the dynamic programming algorithm runs in time $n_e^{O(m)}$.  Combining the running time for the dynamic programming for all cells in all shifted bounding boxes, we get that the algorithm runs in time
\[
\sum_f \sum_{e\in cell(Q_f)} m\cdot n_e^{O(m)} = m^2 h^{O(m)} = m^2 n^{O(m)} = (1/\epsilon^2) n^{O(1/\epsilon)},
\]
as claimed.


\section{Bounds on the number of sensors}
\label{sect:proofs}

Our key result needed for bounding the running time of the dynamic programming in Section~\ref{sect:dp} is the following lemma.

\begin{lemma}
\label{lemma:sensor-bound}
For each cell of size $2mr\times 2mr$, the number of sensors in a square of size $\sqrt{r/2}\times \sqrt{r/2}$ in an optimal solution is at most $O(1) + O(k)=O(k)$.  Moreover, for any vertical strip of size $r\times 2mr$ in a cell, there exists an optimal solution that employs at most $O(m+k)$ sensors in this strip.
\end{lemma}

The following two lemmas are sufficient to establish Lemma~\ref{lemma:sensor-bound}.  Lemma~\ref{lemma:far-sensor-bound} deals with sensors which are far from any stations, i.e., sensors whose moving distances are at least $r/2$.  Lemma~\ref{lemma:close-sensor-bound} considers sensors which stay close to a station.  We remark that we make no attempt to optimize the constants in the proofs.

\begin{lemma}
For each cell of size $2mr\times 2mr$, the number of sensors in a square $q$ of size $\sqrt{r/2}\times \sqrt{r/2}$ in the optimal solution whose moving distances are at least $r/2$ is at most a constant.
\label{lemma:far-sensor-bound}
\end{lemma}
\begin{proof}
Note that sensors in $q$ can cover targets in adjacent squares (all of size $\sqrt{r/2}\times\sqrt{r/2}$). 
However since the sensing radius is at most $r$, they cannot cover targets that much further.  We draw additional 20 adjacent squares as shown in Figure~\ref{fig:adjacent-squares}.  Note that moving a sensor inside $q$ to any of these squares requires moving it for additional distance of $r$; thus increasing the cost by this much.  However, since every sensor considered in this lemma has already moved for the distance of $r/2$, removing one sensor reduces the cost by $r/2$.

\begin{figure}
  \centering
  \begin{tikzpicture}
    \filldraw[gray!20] (-0.5,-0.5) rectangle (+0.5,+0.5);
    \draw[line width=1.5pt] (-0.5,-0.5) rectangle (+0.5,+0.5);
    \draw (-2.5,-1.5) -- (-2.5,+1.5);
    \draw (-1.5,-2.5) -- (-1.5,+2.5);
    \draw (-0.5,-2.5) -- (-0.5,+2.5);
    \draw (+0.5,-2.5) -- (+0.5,+2.5);
    \draw (+1.5,-2.5) -- (+1.5,+2.5);
    \draw (+2.5,-1.5) -- (+2.5,+1.5);
    \draw (-1.5,-2.5) -- (+1.5,-2.5);
    \draw (-2.5,-1.5) -- (+2.5,-1.5);
    \draw (-2.5,-0.5) -- (+2.5,-0.5);
    \draw (-2.5,+0.5) -- (+2.5,+0.5);
    \draw (-2.5,+1.5) -- (+2.5,+1.5);
    \draw (-1.5,+2.5) -- (+1.5,+2.5);
    \node at (0,0) {$q$};
  \end{tikzpicture}
  \caption{Adjacent squares inside which targets can be covered by sensors in $q$}
  \label{fig:adjacent-squares}
\end{figure}
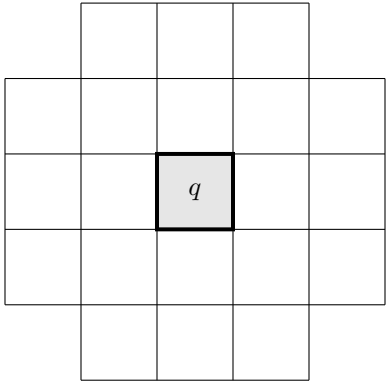

We shall prove the lemma by contradiction.  Let $m_q$ be the number of sensors inside $q$.  Assume that there are at least $63$ sensors in $q$, i.e., $m_q\geq 63$.  Consider another solution where we remove $m_q-21$ sensors, arbitrarily, from $q$, and move another $20$ sensors to all $20$ adjacent squares.  Note that the new schedule covers all previously covered targets.  On the one hand, removing sensors decreases the cost by at least $(m_q-21)\cdot r/2\geq (63-21)\cdot r/2=20r$.  On the other hand, moving sensors to adjacent squares increases the cost by at most $20r$.  Thus, we obtain a cheaper solution and this contradicts the assumption.
\end{proof}

It remains to analyze sensors that are close to the station.   The following Lemma~\ref{lemma:close-sensor-bound} shows that for a station $p\in P$, there exists an optimal solution that schedules at most $O(1)$ sensors within the distance of $r/2$ from $p$.  We again remind that we do not try to optimize the constant in the lemma.

\begin{lemma}
In an optimal solution, the number of sensors whose moving distances are at most $r/2$ from a particular station $p_\ell$ is at most a constant.
\label{lemma:close-sensor-bound}
\end{lemma}

Consider station $p\in P$.  Since it is free to place a sensor at $p$, we assume that all targets within the radius $r$ of $p$ are covered.  We start with an important observation.

\begin{lemma}
  For any distance $0<a\leq r/2$, the number of sensors originating from $p$ whose distance from $p$ is in the range $[a/2,a]$ is at most a constant.
  \label{lemma:ring-bound}
\end{lemma}
\begin{proof}
  Sensors in this range may cover targets $t$ whose distance $d(p,t)$ satisfies $r\leq d(p,t)\leq r+a$.  Call the area in that distance a ring $Ring(r,r+a)$.  Note that we can place $6$ sensors around a circle of radius $2a\leq r$ to cover entirely $Ring(r,r+a)$, paying the cost of $12a$.  Since each sensor whose distance to station $p$ falls in the range $[a/2,a]$ must travel a distance of at least $a/2$, having more than $24$ sensors in this range costs more than this simple feasible solution, leading to a contradiction.
\end{proof}

While the previous lemma is a good step in proving the bound, it can be applied infinitely (as distances can be made smaller and smaller) resulting in an unbounded number of sensors.  To resolve that we need a notion of progress.  We first give an overview for the proof of Lemma~\ref{lemma:close-sensor-bound}.

To prove this lemma, we shall apply Lemma~\ref{lemma:ring-bound} repeatedly to the optimal solution.  If we can show that the number of times we apply the lemma is constant, we know that the total number of sensors is also constant.  Later in this section, we mainly state the required preliminary lemmas; their proofs are deferred to Appendix~\ref{sect:proofs-lemmas}.

Consider the furthest sensor from station $p$ in the optimal solution whose distance from $p$ fall in the range $(0,r/2]$.  Let $a_1$ be the distance and denote that sensor by $s_1$.  Lemma~\ref{lemma:ring-bound} ensures that the number of sensors whose distance from $p$ are in the range of $[a_1/2, a_1]$ is constant.

To apply Lemma~\ref{lemma:ring-bound} again, we look for the furthest sensor which is closer to $p$ than $a_1/2$.  Let $a_2$ be its distance.  We repeat this process, yielding a sequence of distances $a_1, a_2, \ldots$, where
\[
a_{i+1} \leq a_{i}/2.
\]
These distances define a sequence of levels of sensors, where $L_i$ is the set of sensors whose distances from $p$ are in the range $(a_i/2, a_i]$.  

To show that the number of levels is at most a constant, we define the following notion of covered angles to capture the progress.  From the station $p$ located at $(x_p,y_p)$, fix a reference line $\ell$ to be a (half) straight line originating from $(x_p,y_p)$ to $(\infty, y_p)$.  For any target $t$, given $p$ and line $\ell$, we can define the angle $\theta_t$ to be the angle determined by segment $\overline{pt}$ and line $\ell$.  We say that $t$ is {\em at angle} $\theta_t$ and we also refer to $\theta_t$ as the {\em angle} of $t$.  We say that the interval $[\theta_1,\theta_2]$ of angles is {\em covered at radius} $a'$ if every target at angle $\theta\in [\theta_1,\theta_2]$ at distance between $r$ and $a'$ is covered.    

To be precise, we have to describe how to account for covered angles.  When dealing with sensors at level $L_i$, we consider a single sensor $s$ in that level furthest from $p$.  Let $a$ be the distance from $s$ to station $p$.  At this level we consider the interval of angles covered by $p$ at radius $a/2$.  When we consider the next level $L_{i+1}$, we again consider a single sensor $s'\in L_{i+1}$ at furthest distance $a'$ from $p$ and try to account for covered angles at radius $a'/2$.  Since previously covered intervals at radius $a/2$ remain covered at radius $a'/2 < a/2$; we can take all covered intervals from the previous level to this level as well.  We note that while sensors at lower levels $s$ may cover more angles at radius $a'/2$ (for higher level analysis), we keeps the covered intervals fixed when we analyze the additional covered angles; thus the intervals of covered angles act as lowerbounds on the actual angles covered.

Thus, to account for the progress, we maintain a set $S$ of (circular) disjoint intervals of angles and the current guarantee radius $\alpha$, with the condition that any target $t$ with distance at most $\alpha$ from station $p$ at angle $\theta_t$ such that $\theta_t\in I$ for some interval $I\in S$ is covered by the current set of sensors.  
We remark that if this condition holds for some set $S$ with radius $\alpha$, it also holds for any radius $\alpha' < \alpha$.

The following lemma (whose proof is in Appendix~\ref{sect:proofs-lemma-one-circle-coverage}) shows that if we place a sensor at distance $a$ from $p$, it would cover an interval of angles at radius $a'$ of width at least $2\cdot\arccos(13/20)$, where $a' \leq a/2$.  (Also, see Figure~\ref{fig:one-cover-interval}.)

\begin{restatable}{lemma}{onecirclecoverage}
Assume that station $p$ is at $(0,0)$.
If we place sensor $s$ at distance $a$ from $p$, w.l.o.g. at $(a,0)$, where $a\leq r/2$,  the interval 
\[
[-\arccos(13/20),\arccos(13/20)]
\]
is covered at radius $a'$, where $a'\leq a/2$.
\label{lemma:one-circle-coverage}
\end{restatable} 

Lemma~\ref{lemma:one-circle-coverage} essentially allows us to make progress from one level to the next.  However, we need more structures to ensure that when we have circles from many levels.  
Consider two sensors $s$ and $s'$ at two consecutive levels (i.e., $s\in L_i$ and $s'\in L_{i+1}$).  
Assume that $s$ is at distance $a$ from the station $p$ where $a\leq r/2$; this also implies that $s'$ is at distance $a'\leq a/2$ from $p$.
We consider the regions covered additional by $s'$.  Formally, let circles $C_p, C_s,$ and $C_{s'}$ denoted circles of radius $r$ centered at $p,s,$ and $s'$ respectively.  The regions covered additionally by $s'$ is $C_{s'}\setminus (C_p\cup C_s)$, shown in Figure~\ref{fig:two-regions-close}.  The following lemma, regarding these regions, is crucial. 

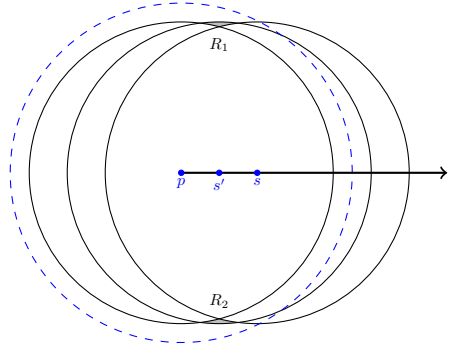
\begin{figure}
\centering
\begin{tikzpicture}[scale=0.5, every node/.style={scale=0.6}]
    \coordinate (C) at (1,0);

    \fill[color=lightgray] (C) circle (4);
    \fill[color=white] (0,0) circle (4);
    \fill[color=white] (2,0) circle (4);

    \draw[black, thin] (0,0) circle (4);
    \draw[black, thin] (2,0) circle (4);
    \draw[black] (C) circle (4);

    \draw[blue, thin, dashed] (0,0) circle (4.5);
        
    \draw[black,thick,->] (0,0) -- (7,0);
    \filldraw [blue] (0,0) circle (2pt) node[anchor=north] {$p$};
    \filldraw [blue] (2,0) circle (2pt) node[anchor=north] {$s$};
    \filldraw [blue] (C) circle (2pt) node[anchor=north] {$s'$};

    \node at (1,3.4) {$R_1$};
    \node at (1,-3.4) {$R_2$};
\end{tikzpicture}
    
\caption{Two regions $R_1$ and $R_2$ in $C_{s'}\setminus (C_p\cup C_s)$. Blue dashed circle shows the distance $r+a'/2$.}
\label{fig:two-regions-close}
\end{figure}

\begin{restatable}{lemma}{tworegionsclose}
Let $s$ and $s'$ be sensors at levels $L_i$ and $L_{i+1}$ with distances $a$ and $a'$ from station $p$, respectively, where $a'\leq a/2\leq r/4$.
Consider $C_{s'}\setminus (C_p\cup C_s)$.  If the set contains two disjoint regions $R_1$ and $R_2$, the further points in both regions is at distance at most $r+a'/2$ from station $p$.
\label{lemma:two-regions-close}
\end{restatable}



The proof of Lemma~\ref{lemma:two-regions-close} appears in Appendix~\ref{sect:geo-lemmas}.  
Given Lemma~\ref{lemma:two-regions-close}, we can argue that if $C_{s'}\setminus (C_p \cup C_s)$ contains two regions, it is cheaper to place two sensor at distance strictly less than $a'/2$ from $p$; as the cost would be $2\cdot a'/2 < a'$; contradicting the optimality of the optimal solution.  
Therefore, the additional covering area of $C_{s'}$ contains only one region.  In that case, we can definitely rotate $C_{s'}$ to maximize the additional covering angle at distance $a'/2$ to be at least $\arccos (13/20)$, half of the angle guaranteed by Lemma~\ref{lemma:one-circle-coverage}.  

Figure~\ref{fig:two-cover-intervals} shows additional angles covered by $s'$.

    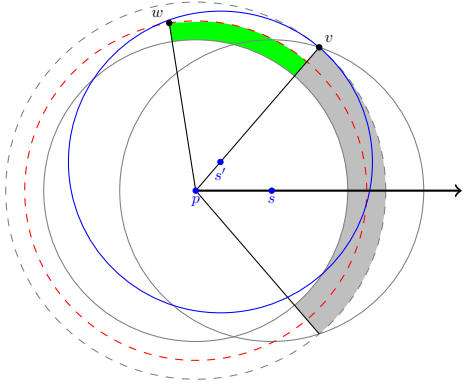
\begin{figure}
    \centering
    \begin{tikzpicture}[scale=0.5, every node/.style={scale=0.6}]
        \coordinate (C) at (0.65, 0.7599342076785331);

        \begin{scope}
            \clip (0,0) circle (4.5);
            \clip (0,0) -- (-0.6975, 4.4456) -- (13,5) -- (13,0) -- (0,0);
            \fill[color=green] (-1,0) rectangle (10,10);
        \end{scope}

        \begin{scope}
            \clip (0,0) circle (5);
            \clip (0,0) -- (6.5,7.6) -- (13,0) -- (6.5,-7.6) -- (0,0);
            \fill[color=lightgray] (0,4) rectangle (10,-4);
        \end{scope}

        \fill[color=white] (0,0) circle (4);

        \draw[gray] (0,0) circle (4);
        \draw[gray] (2,0) circle (4);
        \draw[gray, dashed] (0,0) circle (5);
        \draw[red, dashed] (0,0) circle (4.5);

        \draw[blue] (C) circle (4);

        \draw[black,thick,->] (0,0) -- (7,0);
        \draw[black] (0,0) -- (3.25,3.8);
        \draw[black] (0,0) -- (3.25,-3.8);
        \draw[black] (0,0) -- (-0.6975, 4.4456);

        \filldraw [blue] (0,0) circle (2pt) node[anchor=north] {$p$};
        \filldraw [blue] (2,0) circle (2pt) node[anchor=north] {$s$};
        \filldraw [blue] (C) circle (2pt) node[anchor=north] {$s'$};
        \filldraw [black] (3.25, 3.8) circle (2pt) node[anchor=south west] {$v$};
        \filldraw [black] (-0.6975, 4.4456) circle (2pt) node[anchor=south east] {$w$};
        
    \end{tikzpicture}
    
    \caption{
    How sensor $s'$ (level $L_{i+1}$) at distance $a'=a/2$ from station $p$ covers additional angle interval $[2\theta,\theta]$ at radius $a'/2$ (shown in green), where $\theta = \arccos(13/20)$.  
    }
    \label{fig:two-cover-intervals}
    \end{figure}

Let us state the formal statement of the discussion above regarding two sensors at two consecutive levels.  Its proof appears in Appendix~\ref{sect:proof-of-lemma-additional-angles}.

\begin{restatable}{lemma}{additionalangles}
Consider two sensor $s$ and $s'$ at consecutive levels $L_i$ and $L_{i+1}$.  Let $a$ and $a'$ be the distances from $s$ and $s'$ to $p$.  Let the interval of angles covered by $s$ at radius $a/2$ be $I$.  The additional angle (w.r.t. $I$) covered by $s'$ at radius $a'/2$ is at least $\arccos (13/20)$.
\label{lemma:additional-angles}
\end{restatable}

We are ready to prove our key Lemma~\ref{lemma:close-sensor-bound}

\begin{proof} (of Lemma~\ref{lemma:close-sensor-bound})

We show that the number of levels is constant.  In our analysis we maintain a set $S$ of disjoint covered angle intervals and the radius $b$.  

Recall that $a_1,a_2,\ldots,$ are the furthest distances of sensors $s_1,s_2,\ldots,$ at levels $L_1,L_2,\ldots$ from station $p$.  Initially, starting at level $L_1$, considering the furthest sensor $s_1\in L_1$, from Lemma~\ref{lemma:one-circle-coverage}, we have that at radius $b=a_1/2$, $s_1$ covers the interval of width $2\cdot\arccos(13/20)$.
Thus, we have that $S$ contains a single interval of width $2\cdot\arccos(13/20)$ of covered angles at radius $b=a_1/2$.  

As we consider sensors at higher levels, we make progress by either increase the total width of the intervals or decrease the number of intervals.  Intuitively, we show that while iterating over these levels, one of these events occurs: 
\begin{itemize}
    \item[(1)] when considering a sensor $s_i$ in level $L_i$ with distance $a_i$, the sum of the width of the covered intervals at radius $a_i/2$ increases by at least a constant (more precisely at least $\arccos (13/20)\approx 0.863$), or
    \item[(2)] if $s_i$ does not satisfy the previous property, it means that $s_i$ is ``between'' two other sensors from the lower levels which are well-separated; in this case the show that the number of covered intervals decreases.
\end{itemize}

Consider sensor $s_i$ at distance $a_i$ at level $L_i$.  We consider the intervals of covered angles at radius $b=a_i/2$.  As discussed before every target whose angle is in any interval in $S$ is covered at the new lower radius $b$.  Let interval $J_i$ be the interval of angles covered by $s_i$ at radius $a_i/2$.  Again from Lemma~\ref{lemma:one-circle-coverage}, we have that $J_i$ is of width $2\cdot\arccos(13/20)$.   

We would like to add $J_i$ to $S$.  If $J_i$ only intersects at most one interval $J'$ already in $S$, we know that 
either the additional angle is at least $\arccos(13/20)$, or from Lemma~\ref{lemma:additional-angles}, we can move $s_i$ along the circle of radius $a_i$ so that the additional angle after combining $J_i$ with $J'$ is at least $\arccos(13/20)$; resulting in event (1).  The only case this is not true is when after moving $s_i$, $J_i$ intersects two intervals; thus, adding $J_i$ decreases the number of intervals in $S$ and event (2) occurs.

Since the number of intervals cannot be less than 1, and each time event (1) occurs we increase the total width by a constant and introduce at most one new interval, we know that both events occurs for a constant number of times as well, because the total width is at most $2\pi$.  (To be precise, the number of levels is at most $\lceil 2\cdot 2\pi/\arccos(13/20)\rceil=15$.)  This implies that the number of levels we encounter is at most constant, yielding the lemma.
\end{proof}

\section{Conclusions}

We present a polynomial-time approximation scheme for finding a $(1+\epsilon)$-approximate solution running in time $n^{O(1/\epsilon)}$ for this problem when $k$, the number of base stations, is constant.  This improves, exponentially, over the PTAS presented by Gao~{\em et al.}~\cite{GaoCWC-journalTNET} that runs in time $n^{O(1/\epsilon^2)}$.
We prove a stronger bound on the number of sensors in any unit area in the optimal solution and employ a more refined dynamic programming algorithm whose complexity depends only on the width of the problem.
We hope that the dynamic programming approach presented here might find more applications in other networking problems.




\bibliographystyle{alpha}
\bibliography{coverage}

\appendix


\section{A Counter Example to the Sensor Bound in Gao~{\em et al.}\cite{GaoCWC-journalTNET}}
\label{sect:counter}

The algorithm proposed by Gao~{\em et al.}~\cite{GaoCWC-journalTNET} (first appeared as~\cite{ChenGWC_INFOCOM16}) relies on the fact that the number of sensors in the optimal solution required for each cell $e$ of size $2mr\times 2mr$ is at most $O(m^2)$.  Note that this bound is independent of $k$, the number of stations.  The proof presented in the paper uses the fact that there exists a feasible solution with that many sensors.  However, the actual cost of $k$-MMTC is the total distance, not the number of sensors.
In this section, we provide a counter example to the sensor bound in Theorem 5 in~\cite{GaoCWC-journalTNET} that requires $k$ sensors.

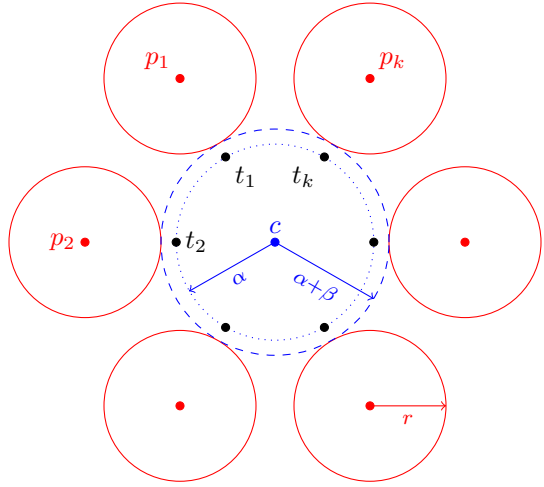
\begin{figure}
  \centering
    \begin{tikzpicture}
    \coordinate (C) at (0,0);
    \coordinate (T1) at (-0.65,+1.13);
    \coordinate (T2) at (-1.30,0);
    \coordinate (T3) at (-0.65,-1.13);
    \coordinate (T4) at (+0.65,-1.13);
    \coordinate (T5) at (+1.30,0);
    \coordinate (T6) at (+0.65,+1.13);
    \coordinate (P1) at (-1.25,+2.17);
    \coordinate (P2) at (-2.50,0);
    \coordinate (P3) at (-1.25,-2.17);
    \coordinate (P4) at (+1.25,-2.17);
    \coordinate (P5) at (+2.50,0);
    \coordinate (P6) at (+1.25,+2.17);
    \draw[red] (P1) circle (1);
    \draw[red] (P2) circle (1);
    \draw[red] (P3) circle (1);
    \draw[red] (P4) circle (1);
    \draw[red] (P5) circle (1);
    \draw[red] (P6) circle (1);
    \draw[blue, dotted] (C) circle (1.3);
    \draw[blue, dashed] (C) circle (1.5);
    \draw[->, blue] (C) -- node[anchor=north, rotate=+30] {\footnotesize $\alpha$} (-1.13,-0.65);
    \draw[->, blue] (C) -- node[anchor=north, rotate=-30] {\footnotesize $\alpha{+}\beta$} (+1.30,-0.75);
    \draw[->, red] (P4) -- node[anchor=north] {\footnotesize $r$} (+2.25,-2.17);
    \filldraw[blue] (C) circle (1.5pt) node[anchor=south] {$c$};
    \filldraw[black] (T1) circle (1.5pt) node[anchor=north west] {$t_1$};
    \filldraw[black] (T2) circle (1.5pt) node[anchor=west] {$t_2$};
    \filldraw[black] (T3) circle (1.5pt);
    \filldraw[black] (T4) circle (1.5pt);
    \filldraw[black] (T5) circle (1.5pt);
    \filldraw[black] (T6) circle (1.5pt) node[anchor=north east] {$t_k$};
    \filldraw[red] (P1) circle (1.5pt) node[anchor=south east] {$p_1$};
    \filldraw[red] (P2) circle (1.5pt) node[anchor=east] {$p_2$};
    \filldraw[red] (P3) circle (1.5pt);
    \filldraw[red] (P4) circle (1.5pt);
    \filldraw[red] (P5) circle (1.5pt);
    \filldraw[red] (P6) circle (1.5pt) node[anchor=south west] {$p_k$};
  \end{tikzpicture}
  \caption{A counter example that requires $k$ sensors.}
  \label{fig:chen-counter}
\end{figure}

We assume that $m$ is large enough, which is true when the error requirement $\epsilon$ is required to be small enough as $m=\Omega(1/\epsilon)$.  Let $c$ be the point at the center of cell $e$.  For any $\alpha>0$ and very small $\beta>0$, place $k$ targets uniformly at distance $\alpha$ from $c$ and for each target $t_i$ place a station $p_i$ at distance $r+\alpha+\beta$ from $c$ in the same direction from $c$ as $t_i$.  Figure~\ref{fig:chen-counter} illustrates this example.  Consider a solution that for each target $t_i$ we move a sensor from $p_i$ with distance $\beta$ to cover it, resulting in the cost of $k\beta$ with $k$ sensors.  If one wants to use smaller number of sensors, some sensor is required to moved further with distance that depends on $\alpha$ and $k$; thus we can always choose $\beta$ to be small enough so that that extra distance is larger than $k\beta$.  This ensures that the optimal solution needs at least $k$ sensors as claimed. 

\section{Proofs of lemmas}
\label{sect:proofs-lemmas}

\subsection{Proof of Lemma~\ref{lemma:one-circle-coverage}}
\label{sect:proofs-lemma-one-circle-coverage}

\onecirclecoverage*

\begin{proof}

    Let $s$ the sensor at distance $a$ from $p$.  Consider two circles $C_p$ and $C_s$ centered at $p$ and $s$ with radius $r$.  
    Let $C'$ be another circle centered at $p$ of radius $r+a'\leq r+a/2$.  Since targets within distance $r+a'$ from $p$ in $C_s$ and $C_p$ are covered, 
    an uncovered target must be in $C'\setminus (C_s \cup C_p)$.  The worst case would attain when $C'$ is largest, so we assume that $a'=a/2$.
    Also, note that the covered interval would be smallest when $a$ is largest; thus, we also assume that $a=r/2$.  This case is shown in Figure~\ref{fig:one-cover-interval}.

    \begin{figure}
    \centering
    \begin{tikzpicture}[scale=0.5, every node/.style={scale=0.6}]
        \begin{scope}
            \clip (0,0) circle (5);
            \clip (0,0) -- (6.5,7.6) -- (13,0) -- (6.5,-7.6) -- (0,0);
            \fill[color=lightgray] (0,4) rectangle (10,-4);
        \end{scope}

        \fill[color=white] (0,0) circle (4);

        \draw[black] (0,0) circle (4);
        \draw[black] (2,0) circle (4);
        \draw[red, dashed] (0,0) circle (5);
        
        \draw[black,thick,->] (0,0) -- (7,0);
        \draw[black] (0,0) -- node[pos=0.5, sloped, above] {$r+a'$} (3.25,3.8);
        \draw[black] (0,0) -- (3.25,-3.8);
        
        \draw[blue, dashed] (2,0) -- node[black, pos=0.5, sloped, above] {$r$} (3.25,3.8);
        \draw[black] (0,0) -- node[pos=0.5, below] {$a$} (2,0);
        
        \draw[blue, dotted] (3.25,0) node[anchor=north] {$q$} -- (3.25,3.8) ;
        \draw[blue, thin] (3.1,0) -- (3.25,0) -- (3.25,0.15) -- (3.1,0.15) -- (3.1,0);

        \filldraw [blue] (0,0) circle (2pt) node[anchor=north] {$p$};
        \filldraw [blue] (2,0) circle (2pt) node[anchor=north] {$s$};
        \filldraw [black] (3.25, 3.8) circle (2pt) node[anchor=south west] {$v$};
    \end{tikzpicture}
    
    \caption{
    How a single sensor at $(a,0)$ covers targets in angle interval $[-\theta,\theta]$.  
    The black circles are $C_p$ and $C_s$; the dashed circle is $C'$.
    Covered area considered in the lemma is shaded.
    }
    \label{fig:one-cover-interval}
    \end{figure}
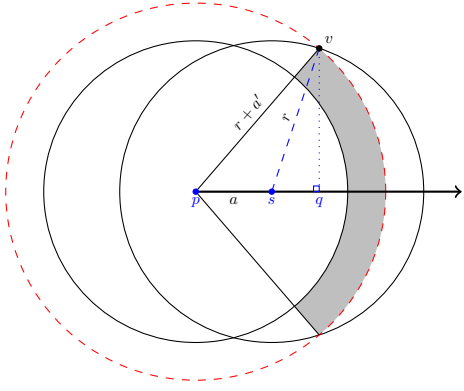
  
    We calculate the angle $\theta$ of $vps$ which is half the angle covered.  
    From the assumptions, we know that the length of $\overline{pv}$ is $r + a'=r+r/4=5r/4$, the length of $\overline{sv}$ is $r$, and the length of $\overline{ps}=a=r/2$.
    Let $q$ be the point on $\ell$ such that angle $pqv$ is the right angle.  Elementary calculation shows that the length of $\overline{sq}$ = $5r/16$.  Thus, the angle $\theta$ equals
    \[
    \arccos \left(\frac{r/2+5r/16}{5r/4}\right)
    =\arccos (13/20).
    \]
    Since the covered interval is $[-\theta,\theta]=[-\arccos (13/20),\arccos (13/20)]$, the lemma follows.
\end{proof}

\subsection{Proof of Lemma~\ref{lemma:two-regions-close}}
\label{sect:geo-lemmas}

\tworegionsclose*

\begin{proof}

To prove this lemma, we assume that $C_{s'}\setminus(C_p\cup C_s)$
contains two regions $R_1$ and $R_2$ and there is some point $t_1$ in
the regions at distance at least $r+a'/2$ from station $p$.  Without
loss of generality, we assume that region $R_1$ is the region of
points with positive $y$ co-ordinates and $R_2$ is the region of
points withs negative $y$ co-ordinates.  We further assume that
$t_1\in R_1$.

Let $\delta=a'/2\leq a/4$.  Given this configuration, we note that we
can certainly move $s'$ such that (1) $t_1$ is at the distance exactly
at $r+a'/2=r+\delta$ from $p$ and also lies on the boundary of $C_s$,
and (2) region $R_2$ contains exactly one point $t_2$ at the
intersection of the boundary of $C_p$ and $C_s$.  Furthermore, this can be done without increasing the distance from $s'$ to $p$.
Figure~\ref{fig:t1t2assumption} illustrates this assumption where
$\delta=a'/2$.  We note that $s'$ is not shown.

\begin{figure}[h]
  \centering
  \begin{tikzpicture}
    \coordinate (P) at (0,0);
    \coordinate (S) at (1,0);
    \coordinate (T1) at (1.03,+2.00);
    \coordinate (T2) at (0.50,-1.94);
    \draw[gray] (2,0) arc (0:360:2);
    \draw[gray] (3,0) arc (0:360:2);
    \draw[red] (2.53,0.06) arc (0:360:2);
    \draw (P) -- (T1) node[pos=0.5, anchor=east] {\small $r{+}\delta$};
    \draw (P) -- (T2) node[pos=0.5, anchor=east] {\small $r$};
    \draw (S) -- (T1) node[pos=0.5, anchor=west] {\small $r$};
    \draw (S) -- (T2) node[pos=0.5, anchor=west] {\small $r$};
    \draw[dashed] (P) -- (S) node[pos=0.5, anchor=north] {\small $a$};
    \filldraw[black] (P) circle (1.5pt) node[anchor=east] {$p$};
    \filldraw[black] (S) circle (1.5pt) node[anchor=west] {$s$};
    \filldraw[black] (T1) circle (1.5pt) node[anchor=south] {\footnotesize $t_1$};
    \filldraw[black] (T2) circle (1.5pt) node[anchor=north] {\footnotesize $t_2$};
  \end{tikzpicture}
  \caption{Points $t_1\in R_1$ and $t_2\in R_2$ in our proofs.}
  \label{fig:t1t2assumption}
\end{figure}
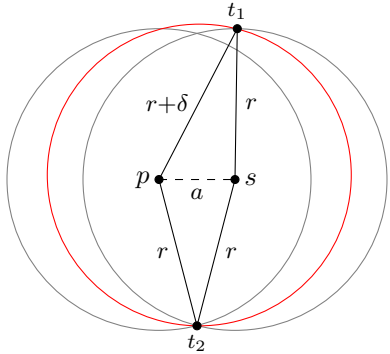

Our goal is to show that for any point $s'$ satisfying the assumption,
$s'$ must be too far from $p$; more precisely,
\[
d(p,s')>2\delta=a';
\]
thus, we have the contradiction as required.

Throughout this section, we consider a simple quadrilateral $\square
pt_2st_1$ (see Figure~\ref{fig:t1t2assumption}).  In our analysis, we
let $p$ be at location $(0,0)$ and $s$ be at $(a,0)$.
If we fix $t_1$ and $t_2$, the point $s'$ must be in $C_{t_1}\cup
C_{t_2}$, to be refered to later on as the feasible region $R_f =
C_{t_1}\cup C_{t_2}$.  Without loss of generality, we assume that $s'$
be the point in $R_f$ closest to $p$. (See
Figure~\ref{fig:feasible-region}.)  The following lemma states the
location of $s'$.

\begin{lemma}\label{lemma:intersect-coordinate}
  Point $s'$ is at 
  \[
  \left(
    \frac{\delta^2 + 2r\delta + a^2}{2a} - \frac{a}2,
    \frac{\sqrt{((2r+\delta)^2 - a^2)(a^2 - \delta^2)}}{2a} - \frac{\sqrt{(2r)^2 - a^2}}2
  \right).
  \]
\end{lemma}
\begin{proof}
  The location for $t_2$ is at
  \[
  \left( \frac{a}2, -\frac{\sqrt{(2r)^2 - a^2}}2 \right),
  \]
  and the location for $t_1$ is at
  \[
  \left( \frac{\delta^2 + 2r\delta + a^2}{2a},
  \frac{\sqrt{((2r+\delta)^2 - a^2)(a^2 - \delta^2)}}{2a} \right).
  \]
  Let $m$ be a midpoint between $t_2$ and $t_1$, the coordinate for
  $m$ is the average of coordinates for $t_1$ and $t_2$.
  We remark that $s\in R_f$.  Furthermore, $s$ is the furthest point
  in $R_f$ from $p$.  Thus, the point $s'$ is the reflection of $s$ on
  the point $m$; its coordinate can be calculated to be as stated
  given that $s$ is at $(a, 0)$.
\end{proof}

\begin{figure}[h]
  \centering
  \begin{tikzpicture}
    \coordinate (P) at (0,0);
    \coordinate (S) at (2,0);
    \coordinate (T1) at (1.73,+2.99);
    \coordinate (T2) at (1.00,-2.83);
    \coordinate (Ss) at (0.73,+0.16);
    \filldraw[gray!20] (T1) -- (1.30,2.26) arc (239.79:275.25:0.84) -- (T1);
    \draw (1.30,2.26) arc (239.79:275.25:0.84);
    \draw (1.59,2.39) node {$\theta$};
    \filldraw[blue, opacity=0.25] (S) arc (70.53:95.26:3) arc (250.53:275.25:3);
    \draw[blue!80] (1.42,0.35) node {$R_f$};
    \draw[dotted] (2.50,-0.23) arc (59.95:115.91:3);
    \draw[dotted] (0.10,+0.47) arc (237.20:285.58:3);
    \draw (P) -- (T1) node[pos=0.5, anchor=east] {\small $r{+}\delta$};
    \draw (P) -- (T2) node[pos=0.5, anchor=east] {\small $r$};
    \draw (S) -- (T1) node[pos=0.5, anchor=west] {\small $r$};
    \draw (S) -- (T2) node[pos=0.5, anchor=west] {\small $r$};
    \draw[dashed] (P) -- (S) node[pos=0.5, anchor=north] {\small $a$};
    \filldraw[black] (P) circle (2pt) node[anchor=east] {$p$};
    \filldraw[black] (S) circle (2pt) node[anchor=west] {$s$};
    \filldraw[black] (T1) circle (2pt) node[anchor=west] {$t_1$};
    \filldraw[black] (T2) circle (2pt) node[anchor=west] {$t_2$};
    \filldraw[black] (Ss) circle (2pt) node[anchor=south] {$s'$};
  \end{tikzpicture}
  \caption{The feasible region $R_f$ and $s'$.}
  \label{fig:feasible-region}
\end{figure}
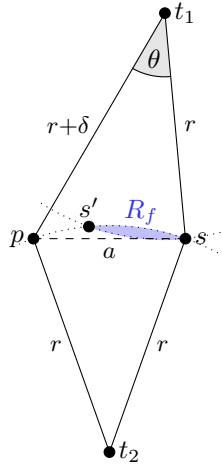

With Lemma~\ref{lemma:intersect-coordinate}, considering only the $x$
coordinates of $s'$, we have that
\begin{align*}
  \frac{\delta^2+2r\delta+a^2}{2a}-\frac{a}{2}
  &= \frac{\delta^2+2r\delta+a^2-a^2}{2a} \\
  &= \frac{\delta^2+2r\delta}{2a} \\
  &\geq \frac{\delta^2+2r\delta}{r} \\
  &= \frac{\delta^2}{r}+2\delta \\
  &> 2\delta,
\end{align*}
where the first inequality follows from the fact that $a\leq r/2$
and the second strict inequality is true because $\delta>0$.
Thus, $d(p,s')> 2\delta$ as required.

\end{proof}

\subsection{Proof of Lemma~\ref{lemma:additional-angles}}

\label{sect:proof-of-lemma-additional-angles}

\additionalangles*

\begin{proof}
Let $C_s$ and $C_{s'}$ be two coverage circles for sensor $s$ and $s'$, and $C_p$ be the circle of radius $r$ centered at $p$.
Lemma~\ref{lemma:two-regions-close} implies that if the area $C_{s'}\setminus (C_s\cup C_p)$ contains two regions, they can be covered by two circles located at distance at most $a'/2$ from $p$, with the total cost of less than $a'$; but this leads to a contradiction.

Thus, $C_{s'}\setminus (C_s\cup C_p)$ contains only on region. 
Let $v$ be the intersection of $C_{s'}$ and $C_s$ outside $C_p$.  (See Figure~\ref{fig:two-cover-intervals}.)
We claim that the additional angle covered by $C_{s'}$ is at least $\arccos(13/20)$.  To see this, assume otherwise.    In this case, the distance from $v$ to $p$ is less than $r+a'$; thus we can move $s'$ along the circle of radius $a'$ centered at $p$ with the same movement cost while increasing the coverage area outside $C_s\cup C_p$.
\end{proof}


\end{document}